\documentclass[10pt,conference]{IEEEtran}
\normalsize
\usepackage{amsmath,amsfonts,amssymb,amsbsy,url,verbatim,amsthm}
\usepackage{comment}
\usepackage{graphicx}
\usepackage{balance}
\usepackage{times}
\usepackage{dsfont}
\usepackage{subfig}
\usepackage{pgfplots}
\usepackage{cite}
\usepackage{bm}
\usepackage{algorithm,algorithmic, color}

\graphicspath{{fig/}}

\makeatletter
\def\input@path{{fig/}{./}}
\makeatother

\newtheorem{proposition}{Proposition}

\newcommand{\hermitian}{^{H\scriptscriptstyle}}
\newcommand{\transpose}{^{T\scriptscriptstyle}}

\newcommand{\bi}{\begin{itemize}}
\newcommand{\ei}{\end{itemize}}
\newcommand{\ben}{\begin{enumerate}}
\newcommand{\een}{\end{enumerate}}
\newcommand{\bc}{\begin{cases}}
\newcommand{\ec}{\end{cases}}
\newcommand{\bd}{\begin{description}}
\newcommand{\ed}{\end{description}}
\newcommand{\e}{\item}
\newcommand{\be}{\begin{equation}}
\newcommand{\ee}{\end{equation}}
\newcommand{\bea}{\begin{eqnarray}}
\newcommand{\eea}{\end{eqnarray}}

\def \x{\mathbf x}

\def \z{\mathbf z}
\def \y{\mathbf y}
\def \t{\mathbf t}
\def \e{\mathbf e}
\def \f{\mathbf f}
\def \b{\mathbf b}
\def \a{\mathbf a}
\def \u{\mathbf u}

\def \M{\mathbf M}

\IEEEoverridecommandlockouts

\begin{document}

\title{Cube-Split: Structured Quantizers on the Grassmannian of Lines}

\author{Alexis~Decurninge, Maxime~Guillaud\\
Mathematical and Algorithmic Sciences Laboratory, Huawei Technologies\\
France Research Center, 20 quai du Point du Jour, 92100 Boulogne Billancourt, France\\
email: \texttt{\{alexis.decurninge,maxime.guillaud\}@huawei.com}}

\maketitle

\begin{abstract}
This paper introduces a new quantization scheme for real and complex Grassmannian sources. The proposed approach relies on a structured codebook based on a geometric construction of a collection of bent grids defined from an initial mesh on the unit-norm sphere. The associated encoding and decoding algorithms have very low complexity (equivalent to a scalar quantizer), while their efficiency (in terms of the achieved distortion) is on par with the best known structured approaches, and compares well with the theoretical bounds.
These properties make this codebook suitable for high-resolutions, real-time applications such as channel state feedback in massive multiple-input multiple-output (MIMO) wireless communication systems. 
\end{abstract}


\section{Introduction}
We address in this paper the quantization of a source uniformly distributed on a real or complex Grassmannian. An important performance metric for a quantizer is the resulting average quantization error, or distortion.
For a given number of quantization bits, theoretical limits on the minimum attainable distortion have been drawn (see e.g \cite{dai08}). In order to attain this bound, some quantizer design strategies which have been proposed use numerically optimized codebooks (e.g. without particular structure \cite{Linde_Buzo_Gray_algo_TCOM1980} or incorporating additional structure in the searched codebook in order to lower the optimization complexity \cite{hochwald00}). Unfortunately, the codebook size must increase exponentially with the Grassmannian dimension in order to maintain a given effective quantization accuracy \cite{loveheath03}. Thus, since these codebooks require an exhaustive search and large storage capacity, they are intractable for large dimensions.

More recently, emphasis has been put on high-dimensional and high-resolution quantizers with less complex encoders and decoders, while allowing some departure from the theoretical optimal distortion bounds. In \cite{boccardi07}, a fast quantization based on a codebook comprised of Fourier vectors is introduced. However, the Fourier structure results in poor packing efficiency (we define efficiency as the ability to achieve the optimal slope in the distortion vs. codebook size performance plots at high resolution, i.e. for asymptotically large codebooks). The method of \cite{choi15} using a trellis structure to design codebooks suffers from the same drawback. Codebooks based on projecting a lattice onto a sphere (see \cite{ryan07}) or on a lattice combined with simplices \cite{ryan09} have good efficiency. However, in these approaches, the quantizer is still exponentially complex with respect to the number of bits while exhibiting a polynomial complexity with respect to the dimension of the source.

Product quantizers \cite[p.~430]{gersho_gray_vector_quantization_book} constitute another class of vector quantizers based on a decomposition or a transformation of the initial source into one or multiple components (usually lower-dimensional) that should be easier to quantize. In the unidimensional case, Bennett proposed so-called companders \cite{bennett48} asymptotically minimizing the distortion, which have been generalized for vector quantization e.g. in \cite{na95,zheng08}. 

The Grassmannian quantization problem has important applications in Massive MIMO (Multiple-Input Multiple-Output) wireless communication systems. For example, accurate channel state information (CSI) is required in order to obtain the large multiplexing gain expected from massive MIMO systems \cite{JSAC2013HoydisDebbah}. In this context, CSI needs to be known up to a multiplication by a complex value, which corresponds to a point on the Grassmannian space of complex lines; hence Grassmannian representations are well adapted to such applications, as demonstrated by the abundant literature on this topic (see  \cite{lovelaurao08} and references therein).\\

We propose in the following a quantizer based on companders for a vector uniformly distributed on a real or complex Grassmannian space, with the following properties:
\begin{itemize}
\item zero storage requirements,
\item the complexity of encoding/decoding is equivalent to the complexity of a scalar quantizer,
\item packing efficiency is comparable to the best state of the art structured approaches.
\end{itemize}

We start by introducing general notations on the considered Grassmannian space in Section \ref{sec:def} and the quantizer design problem in Section \ref{section_quantizer}. Sections \ref{section_cubesplit1} and \ref{section_cubesplit2} introduce a new \emph{cube split} quantizer for both real and complex Grassmannian sources. We finally present simulated performances results in Section \ref{section_simulations}.

\section{Definition and notations}
\label{sec:def}
Let $\mathbb{K}^d$ denote a $d$-dimensional vector space on a field $\mathbb{K}$, where $\mathbb{K}$ can be either $\mathbb{R}$ (real case) or $\mathbb{C}$ (complex case). A Grassmannian space $G(\mathbb{K}^d,r)$ is defined as the space of subspaces of $\mathbb{K}^d$ of dimension $r$. In this article, we focus on the case $r=1$, i.e. the set of lines in $\mathbb{K}^d$. 
We will use vectors to represents elements of $G(\mathbb{K}^d,1)$, i.e. the vector $\x\in\mathbb{K}^d$ represents the set $\{\lambda\x,\lambda\in\mathbb{K}\}$, which is a point in $G(\mathbb{K}^d,1)$. For simplicity, we will use vectors of unit Euclidean norm ($\|\x\|=1$) as representatives of the Grassmannian variable.
Thus we can define the chordal distance between two Grassmannian lines represented by their respective spanning vectors $\x$ and $\y$ as
\[
d_C(\x,\y) = \left\{\begin{array}{ll}
\sqrt{1-\left|\x\hermitian\y\right|^2} &\text{ if $\mathbb{K}=\mathbb{C}$},\\
\sqrt{1-\left|\x\transpose\y\right|^2} &\text{ if $\mathbb{K}=\mathbb{R}$}.
\end{array}\right.
\]\\

\section{Quantizer design problem}
\label{section_quantizer}
Let us consider a quantization function $Q$ defined on the Grassmannian $\mathcal{G}$, characterized by a partition $(C_1,\dots,C_N)$ of $\mathcal{G}$ into $N$ \emph{decision regions} and a codebook $(\x_1,\dots,\x_N)$ of unit-norm vectors in $\mathbb{K}^d$ such that for any input vector $\x$, its quantized version is 
\begin{equation}
Q(\x) = \x_i \text{  if  } \x\in C_i.  \label{eq_decision}
\end{equation}

The distortion of a quantizer is measured as the expected error between the source vector $\x$ and its quantized version $Q(\x)$
\[
\mathbb{E}_{\x}\left[ \varphi(d_C(\x,Q(\x))) \right]
\]
with $\varphi$ an increasing function depending on the considered problem, and where the expectation is taken over the distribution of the source. In this article, we focus on sources uniformly distributed on $\mathcal{G}\triangleq G(\mathbb{K}^d,1)$. Quantizer design problems typically revolve around designing quantizers which operate as close as possible to the minimum distortion for a fixed bit rate; in other words, for a given codebook size $N$ and a given source distribution, the choice of the decision regions $C_i$ and the corresponding reconstruction codewords $\x_i$ should minimize the distortion. The properties of the optimal quantizer that an efficient quantizer should ideally fulfill are \cite{gersho_gray_vector_quantization_book}:
\begin{enumerate}
\item For a given codebook, the optimal decision regions $C_i$ correspond to the collection of \emph{Voronoi regions} $V_i$ defined as 
\begin{equation} \label{eq_def_Voronoi}
V_i = \left\{ \x\in \mathcal{G}  : d_C(\x,\x_i) \leq d_C(\x,\x_j) \quad \forall j\neq i  \right\}, 
\end{equation}
i.e. the set of points of $\mathcal{G}$ which are closer to $\x_i$ than to any other codeword from the codebook. 
\item For given decision cells, the codewords $\x_i$ must correspond to the ``barycenter'' of the corresponding decision region defined as
\[
\x_i = \arg\min_{\y}\mathbb{E}_{\x\in C_i}\left[ \varphi(d_C(\x,\y ))\right].
\]
\item In the limit of large $N$ (high-resolution quantizers), and for uniform source distribution, the limiting distribution of the codewords $\x_1,\dots,\x_N$ of an optimal quantizer should be uniform as well for a large family of distortion measures \cite{zheng07}.
\end{enumerate}

The design of minimum distortion quantizers for large dimensions $d$ and for large $N$ is a non trivial problem. Numerical approaches (such as the Linde-Buzo-Gray algorithm \cite{Linde_Buzo_Gray_algo_TCOM1980}) have been proposed, which strive to achieve optimal efficiency according to the above criteria.  Since numerically optimized codebooks are generally unstructured, the associated decision regions (usually chosen to coincide with the Voronoi regions, $C_i = V_i$) lack a more tractable definition than the one in \eqref{eq_def_Voronoi}. As a consequence, deciding upon the output of the quantizer (see eq.~\eqref{eq_decision}) involves solving the optimization problem
\begin{equation}
i=\arg\min_{i=1\ldots N} d_C(\x,\x_i),  \label{eq_argmin_quantization}
\end{equation}
for each realization of the random variable $\x$. For large codebooks (e.g. if the index $i$ is encoded using $B=50$ bits, the codebook contains $N=2^B \approx 10^{15}$ vectors), it is clear that both the storage of the codebook and the complexity of evaluating the $N$ distances involved in \eqref{eq_argmin_quantization} are not practical options.

In order to address this issue, we focus instead on designing quantizers (i.e. codebooks and their associated decision regions) whose structure greatly reduces both the codebook storage problem and the computational requirements associated with the quantization operation, at the cost of a slight loss of distortion optimality. In particular, in the rest of the article, we introduce quantizer designs for which the structure of the decision regions enables a very efficient (low-complexity, zero storage) computation of \eqref{eq_decision}, at the cost of a slight relaxation of the optimal efficiency conditions detailed above.


%

\section{Cube Split Quantizer for $G(\mathbb{R}^d,1)$}
\label{section_cubesplit1}

Let us first consider the quantization of $\y$ uniformly distributed on $G(\mathbb{R}^d,1)$. The rationale of the proposed quantizer is the following: first split the considered Grassmannian space (homogeneous to a sphere)  into cells looking like bent hypercubes (hence the \textit{cube-split} name), and then define on each cell a bent lattice through a mapping chosen such that the resulting codewords are approximately uniformly distributed on the sphere. \\
More specifically, we propose an encoder that numerically computes a sequence of bits from $\y$ through the following major steps:
\begin{itemize}
\item \textbf{Step 1}: The determination of an initial cell and its corresponding index as well as its binary representation; this yields the first bits of the codeword index.
\item \textbf{Step 2}: The computation of the remaining bits defining the relative position of the codeword in the initial cell: this relative position is captured by $d-1$ local coordinates which are successively and independently quantized.
\end{itemize}

\subsection{Initial mesh}
Let $\mathbb{d}$ $\e_1=[1,0,\dots,0]^T, \e_2=[0,1,0,\dots,0]^T, \dots, \e_d=[0,\dots,0,1]^T$ denote the elements of the canonical basis in $\mathbb{R}^d$. The first part ($\lceil{\log_2(d)}\rceil$ bits) of the codeword index is the index $i^*$ of the canonical basis vector closest to the source vector, i.e.
\[
i^*  = \arg\min_i d_C(\y,\e_i) = \arg\max_{i}|y_i|.
\]
This operation defines a coarse quantizer, whereby $G(\mathbb{R}^d,1)$ is split into $d$ cells $C_1^0 \ldots C_d^0$, defined as 
\[
C_i^0 =  \left\{\y\in G(\mathbb{R}^d,1)\text{ s.t. } |\y\transpose \e_i|>|\y\transpose \e_j| \text{ for } j\neq i\right\}.
\]
for $1\leq i \leq d$. Note that this results in splitting the real unit-sphere into $\mathbb{R}^d$ into $2d$ cells, since colinear vectors with opposite sign in the real sphere are equivalent in the Grassmannian space (see Fig. \ref{fig:initial_cells} for an illustration in $\mathbb{R}^3$).

Let us note that the choice of the canonical basis to build the initial mesh is arbitrary -- indeed, the method can be generalized to any (possibly overcomplete) basis. However, it allows to define simple and computationally efficient local coordinates on each cell, as will be seen in the next section.
\begin{figure}[h!]
\centering
\includegraphics[width=6cm]{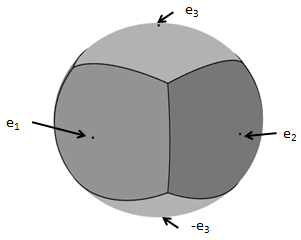}
\caption{Illustration of the initial cells on $G(\mathbb{R}^3,1)$; the three colors denote $C_1^0, C_2^0$ and $C_3^0$.}
\label{fig:initial_cells}
\end{figure}

The remaining part of the codeword index is dedicated to encoding the relative position of $\y$ in $C_{i^*}^0$, as detailed next.\\

\subsection{Companding of local coordinates}

In order to obtain the cube-split quantizer, each cell $C_{i^*}^0$ of the initial mesh is further divided into smaller decision cells by defining a mapping between $C_{i^*}^0$ and the unit $d-1$ dimensional cube, and defining a simple scalar quantization scheme on each coordinate of the cube. The local coordinates on $C_{i^*}^0$ are defined as the elements of the vector
\[
\t =\left(\frac{y_1}{y_{i^*}}, \ldots,\frac{y_{i^*-1}}{y_{i^*}}, \frac{y_{i^*+1}}{y_{i^*}},\ldots \frac{y_{d}}{y_{i^*}}  \right)^T.
\]
It is clear that $\t$ together with $i^*$ is sufficient to uniquely identify a point on $G(\mathbb{R}^d,1)$. Furthermore the distribution of the vector $\t$  (conditionally to the fact that $i^*=\arg\max_i |y_i|$) is known as multivariate Cauchy truncated on $[-1;1]^{d-1}$. Neglecting the statistical dependence between the components of this multivariate distribution (which is hard to compensate due to the truncation), we choose to independently quantize each component. Furthermore, in order to approximately obtain an asymptotically (in the high-resolution regime) uniform distribution of the codewords, we apply a scalar compander to each component, defined as follows. Let $\M_{i^*}$ denote the mapping between $C_{i^*}^0$ and the unit cube by
\begin{equation} \label{real_compander_1}
\begin{array}{cccc}
\M_{i^*}: & C_{i^*}^0 &\rightarrow & [0;1]^{d-1}\\
& \y=(y_1,\dots,y_d)^T & \mapsto & (a_1,\dots,a_{d-1})^T
\end{array}
\end{equation}
with for any $1\leq i\leq d-1$
\begin{equation} 
\label{eq:mapping}
a_i = \frac{2}{\pi}\tan^{-1}\left(t_i\right)+\frac{1}{2}.
\end{equation}
(see Appendix~\ref{section_appendixproofs} for a proof that the proposed mapping results in a distribution of the $a_i$ with uniform marginals on $[0;1]$).
It remains to independently quantize the $a_i$; for this, one defines a uniform scalar quantizer for each coordinate $a_i$, $i=1\ldots d-1$ where the a scalar quantization in $[0;1]$ is performed with $B_i$ bits. In other terms, the sent sequence of bits is the binary representation of
\[
\left\lfloor{2^{B_i}a_i-\frac{1}{2}}\right\rceil
\]
with $\lfloor .\rceil$ denoting the closest integer rounding operator.\\
Finally, the inverse mapping (required at the decoder) can be obtained for $\a = (a_1,\dots,a_{d-1})^T\in[0;1]^{d-1}$ as 
\begin{equation}
\M_{i^*}^{-1}(\a) = \frac{1}{\sqrt{1+\sum_{i=1}^{d-1}u_i^2}}\left(\begin{array}{c} u_1\\ \vdots\\u_{i^*-1}\\ 1\\u_{i^*}\\ \vdots\\u_{d-1}\end{array}\right)
\label{eq:decoder1}
\end{equation}
where $u_i=\tan\left(\frac{\pi}{2}\left(a_i-\frac{1}{2}\right)\right)$ for $1\leq i\leq d-1$.

The codewords $\x_1 \ldots \x_N$ corresponding to a codebook generated using the proposed cube split design are illustrated in Fig. \ref{fig:voronoi2}, together with the corresponding decision regions and Voronoi cells. 
The fact that the proposed cube-split design is only slightly suboptimal according to criteria given in Section~\ref{section_quantizer} can be observed on the figure, and will be confirmed by more extensive simulations in Section~\ref{section_simulations}.




\begin{figure}[h!]
\centering
\includegraphics[width=8cm]{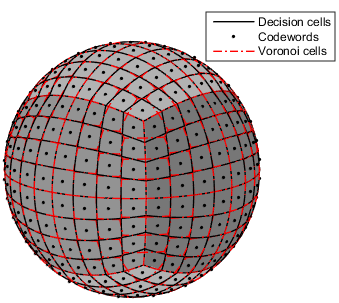}
\caption{Illustration of the cube-split codebook on $G(\mathbb{R}^3,1)$ for $B_1=B_2=3 \ \text{bits}$. The codebook defines $d.2^{(B1+B2)}=192$ lines in $\mathbb{R}^3$, each intersecting twice with the unit sphere.  Encoding of the quantized position on the sphere requires approximately 7.6 bits.} 
\label{fig:voronoi2}
\end{figure}

\section{Cube Split Quantizer for $G(\mathbb{C}^D,1)$}
\label{section_cubesplit2}

We now generalize the cube-split approach to the complex case, i.e. one point in the Grassmannian space is one line in $\mathbb{C}^D$ of the form $\{\lambda\x,\lambda\in\mathbb{C}\}$, represented by $\x\in\mathbb{C}^D$.
We introduce two possible designs for this case.
\subsection{Scheme 1: Use a real representation}
\label{complex_CS_scheme1}

A natural extension of the previously proposed scheme to the complex case consists in treating the real and imaginary components of $\x$ as two independent real dimensions.
For this, we transform the initial complex vector $\x$ into a real vector in the real Grassmannian $G(\mathbb{R}^{2D-1},1)$. Let us first denote
\begin{equation}
\x = \left(\begin{array}{c}x_1\\ \vdots\\ x_D\end{array}\right) \in \mathbb{C}^D.
\end{equation}
By noting $\varphi=\arg(x_1)$, we define a rotated equivalent vector $\x^{(r)}$ as $\x^{(r)}=\x e^{-i\varphi}$. Note that all codewords will have a real first coordinate with this scheme (this is acceptable thanks to the invariance with respect the phase mentioned above); $\x^{(r)}$ can be rewritten as the real vector
\begin{equation}
\label{eq:representative}
\y = \left(\begin{array}{c}\text{Re}(x_1^{(r)})\\ \vdots\\ \text{Re}(x_D^{(r)})\\ \text{Im}(x_2^{(r)})\\ \vdots \\ \text{Im}(x_D^{(r)})\end{array}\right)  \in \mathbb{R}^{2D-1}.
\end{equation}
We then quantize the real vector $\y$ in the real Grassmannian of dimension $d=2D-1$ following the scheme from Section \ref{section_cubesplit1}.\\

Note however that for $\x$ uniformly distributed on $G(\mathbb{C}^D,1)$, the above transformation yields a vector whose first component $\text{Re}(x_1^{(r)})$ has different statistics from the remaining $2D-2$ components due to the statistical dependence of $\varphi$ on $x_1$, i.e. $\y$ is not uniformly distributed on $G(\mathbb{R}^{2D-1},1)$. This hints at a possible suboptimality of this scheme, which will be verified in the simulations (Section~\ref{section_simulations}).



\subsection{Scheme 2: complex Grassmannian Cube Split quantizer}
\label{complex_CS_scheme2}
Since Scheme 1 is dependent on the choice of the real representative, we present an alternative scheme for the quantization of a complex Grassmannian vector $\x$ using different initial mesh and mapping.
\subsubsection{Initial mesh}
In the same spirit as in Scheme 1, the first $\lceil{\log_2(D)}\rceil$ bits of the codeword index represent the index $i^*$ of the closest vector $\f_{i^*}$ amongst the complex canonical basis in $\mathbb{C}^D$ denoted by $\f_1=[1,0,\dots,0]^T, \f_2=[0,1,0,\dots,0]^T, \dots, \f_D=[0,\dots,0,1]^T$. Therefore
\[
i^* =\arg\min_i d_C(\x,\f_i) = \arg\max_{i}|x_i|.
\]
where $|.|$ denotes the complex modulus. The canonical basis induces an initial mesh on the complex Grassmannian $G(\mathbb{C}^D,1)$ with $D$ cells defined as 
\[
C^{\mathbb{C}}_i = \left\{\x\in G(\mathbb{C}^D,1)\text{ such that } |\x\hermitian \f_i|>|\x\hermitian \f_j|\text{ for } j\neq i\right\}.
\]
There remains to define the second part of the codeword index, which encodes the relative position of $\x$ with respect to $\f_{i^*}$.

\subsubsection{Companding of local coordinates}

For the complex case, let us define local coordinates through the vector
\[
\t =\left(\frac{x_1}{x_{i^*}}, \ldots,\frac{x_{i^*-1}}{x_{i^*}}, \frac{x_{i^*+1}}{x_{i^*}},\ldots \frac{x_{D}}{x_{i^*}}  \right)^T.
\]
Again, it is clear that $\t$ together with $i^*$ is sufficient to identify a point in $G(\mathbb{C}^D,1)$. Furthermore, the distribution of $\t$ is known as complex multivariate Cauchy truncated on $\mathcal{D}_1^{D-1}$ where $\mathcal{D}$ denotes the unit complex disk. We introduce the mapping $\M_{i^*}^{\mathbb{C}}$ such that each complex coordinate of $\t$ is uniformly distributed on $[0;1]^2$ if $\x$ is uniformly distributed on $G(\mathbb{C}^D,1)$. Specifically, for each cell $C^{\mathbb{C}}_{i^*}$ of the initial mesh, we define 
\begin{equation}
\label{eq:mapping2_0}
\begin{array}{cccc}
\M^{\mathbb{C}}_{i^*}: & C^{\mathbb{C}}_{i^*} &\rightarrow & [0;1]^{2D-2}\\
& \x=(x_1,\dots,x_D)^T & \mapsto & \a = (a_1,\dots,a_{2D-2})^T
\end{array}
\end{equation}
with for any $1\leq i\leq D-1$, $a_{2i-1} = \mathcal{N}(\text{Re}(w_i))$ and $a_{2i}=\mathcal{N}(\text{Im}(w_i))$ where $\mathcal{N}$ is the cumulative distribution function of the standard real univariate Gaussian, i.e.
\[
\mathcal{N}(x) = \int_{-\infty}^x \frac{1}{\sqrt{2\pi}} e^{-y^2/2}dy
\]
and for $1\leq i\leq D-1$
\begin{equation}
\label{eq:mapping2}
w_i = \sqrt{2}\log\left(\frac{1+|t_i|^2}{1-|t_i|^2}\right)^{1/2}\frac{t_i}{|t_{i}|}.
\end{equation}
For $\x$ uniform on $G(\mathbb{C}^D,1)$, the resulting $\a$ has uniform marginals (see Proposition~\ref{prop2} in Appendix~\ref{section_appendixproofs}) which again prompts the use of independent scalar quantizers for its components.
More precisely, one then defines a regular grid on the cube $[0;1]^{2D-2}$, i.e. for each coordinate $a_1,\dots,a_{2D-2}$, a scalar quantization is performed in $[0;1]$. The inverse mapping (required at the decoder) can be obtained from $\a=(a_1,\dots,a_{2D-2})$ as 
\begin{equation}
(\M^{\mathbb{C}}_{i^*})^{-1}(\a) = \frac{1}{\sqrt{1+\sum_{i=1}^{D-1}|z_i|^2}}\left(\begin{array}{c} z_1\\ \vdots\\z_{i^*-1}\\ 1\\z_{i^*}\\ \vdots\\z_{D-1}\end{array}\right)
\label{eq:decoder2}
\end{equation}
where for all $1\leq i\leq D-1$
\[
z_i=\sqrt{\frac{1-e^{-\frac{|w_i|^2}{2}}}{1+e^{-\frac{|w_i|^2}{2}}}}\frac{w_i}{|w_i|} \text{  and  } w_i=\mathcal{N}^{-1}(a_{2i-1})+i\mathcal{N}^{-1}(a_{2i}).
\]
The encoding (from $\x \in G(\mathbb{C}^D,1)$ to the codeword index) and decoding (from the codeword index to $Q(\x)$) algorithms for this scheme are summarized in Algorithm \ref{algo_cs2}.

\begin{algorithm} 
	\caption{Encoder and decoder of Scheme 2.}
	\label{algo_cs2}
	\begin{algorithmic}

     \STATE{\textbf{Encoder:} Compute the sequence of bits representing the quantized version of a complex vector $\x$ in the Grassmannian $G(\mathbb{C}^D,1)$.}
    \STATE{ }

    \STATE{-- Compute $i^*=\arg\max|x_{i}|$.}
    \STATE{-- Initialize the bits sequence: $\b\leftarrow$ binary representation of $(i^*-1)$ computed with $\lceil{\log_2(D)}\rceil$ bits.}
    \STATE{-- Compute $\a=\M^{\mathbb{C}}_{i^*}(\x)$ (Eq. (\ref{eq:mapping2_0}) and (\ref{eq:mapping2})).}

    \STATE{\textbf{for} $i = 1\dots 2D-2$}
    \STATE{$\quad$ -- Compute the integer representation $\b_i$ of $\left\lfloor{2^{B_i}a_i-\frac{1}{2}}\right\rceil$.}
    \STATE{$\quad$ -- $\b\leftarrow [\b,\b_i]$.}
    \STATE{\textbf{end}}
    \STATE{\textbf{Output:} the  computed sequence of bits $\b$.}
    \STATE{ }

   \STATE{\textbf{Decoder:} Compute the quantized vector represented by a sequence of bits $\b$.}
    \STATE{ }

    \STATE{\textbf{for} $i = 2D-2\dots 1$}
    \STATE{$\quad$ -- Compute $n_i$ as the binary representation of the last $B_i$ bits and remove these bits from $\b$.}
    \STATE{$\quad$ -- Compute $a_i = 2^{-B_i}\left(n_i+\frac{1}{2}\right)$.}
    \STATE{\textbf{end}}
  \STATE{-- Compute $i^*$ from the first $\lceil{\log_2(D)}\rceil$ bits of $\b$.}
   \STATE{-- Compute $\hat{\x}=(\M^{\mathbb{C}}_{i^*})^{-1}(\a)$ (Eq. (\ref{eq:decoder2})).}   
    \STATE{\textbf{Output:} the  quantized vector $\hat{\x}$.}
    	\end{algorithmic}
\end{algorithm}


\section{Simulation results}
\label{section_simulations}
In this section, we perform Monte Carlo simulations to compare the average squared chordal error defined as
\[
D(Q) = \mathbb{E}_{\x}[d_C(\x,Q(\x))^2]
\]
achieved by different quantizers for a uniform source on the complex Grassmannian of lines. Note that this metric has important operational significance in the context of MIMO communications \cite{ryan07,xia06}, among others.
The following codebooks and quantization approaches are considered:
\begin{itemize}
\item Fourier codebooks \cite{boccardi07}
\item Square lattice angular quantization (SLAQ) \cite{ryan07}
\item Scalar quantization
\item Cube split quantizers using real representative (Scheme 1) and defined directly in complex (Scheme 2).
\end{itemize}
In order to give upper and lower bounds on the distortion measure, we exploit the result from \cite[Th. 2]{dai08} on the performance of the best codebook of cardinality $N=2^B$, valid for the high-resolution (large $B$) regime:
\begin{equation}
\label{eq:asymp1}
\frac{d-1}{d}2^{-\frac{B}{d-1}}\leq  \inf_{Q\in\mathcal{Q}_{2^B}} D(Q) \leq  \Gamma\left(\frac{d}{d-1}\right)2^{-\frac{B}{d-1}}
\end{equation}
where $\mathcal{Q}_{N}$ denotes the set of  quantizers with $N$ codewords.\\

The distortion achieved by the various approaches, together with the bounds from eq.~(\ref{eq:asymp1}), are depicted in Figs. \ref{fig:gain} and \ref{fig:gain2} for the case of 4 and 64 complex dimensions respectively. The average number of bits per dimension is computed as $\frac{1}{D}\left(\lceil{\log_2(D)}\rceil+\sum_{i=1}^{2D-2} B_i\right)$.
These results demonstrate that in terms of distortion, the proposed cube-split quantizers perform comparably with SLAQ, while they outperform the other approaches
(unstructured codebooks generated e.g. through the generalized Lloyd algorithm \cite{Linde_Buzo_Gray_algo_TCOM1980} can not practically be simulated for the considered codebook sizes, and therefore are not included in the curves).\\ 

Table \ref{table:feature} presents an overview of the available approaches when also encoder complexity must be taken into account. The efficiency reported in this table corresponds to the capacity to achieve the optimal slope in the (logarithmic) chordal error vs. codebook size curves.
Indeed, it may be observed that the efficiency of Fourier codebooks tends asymptotically to zero because Fourier codewords do not span the complex Grassmannian (since it constrains the modulus of each component of the quantized vector to be equal to one). 
From this table, it can be seen that the cube-split quantization approach is a strong candidate for complexity-constrained applications, thanks to the low complexity of its encoder and decoder algorithms (Algorithm~\ref{algo_cs2}), while the achieved distortion remains close to the theoretical packing bound.

\begin{figure}[h]
\scriptsize 
\centering
%
%
\definecolor{mycolor1}{rgb}{1.00000,0.00000,1.00000}%
\definecolor{mycolor2}{rgb}{0.00000,1.00000,1.00000}%
\begin{tikzpicture}

\begin{axis}[%
width=2.5in,
height=2in,
at={(0.758333in,0.48125in)},
scale only axis,
xmin=0.5,
xmax=5,
xmajorgrids,
ymin=-35,
ymax=0,
ymajorgrids,
axis x line*=bottom,
axis y line*=left,
legend style={at={(0.004405,0.024603)},anchor=south west,legend cell align=left,align=left,draw=white!15!black},
xlabel={Average number of bits/dimension},
ylabel={$D(Q)$ [dB]},
]
\addplot [color=blue,solid,mark=asterisk,mark options={solid}]
  table[row sep=crcr]{%
0.75	-3.15608913134753\\
1	-4.14075082940065\\
1.25	-4.72440704296251\\
1.5	-5.41653145987659\\
1.75	-6.35365660781227\\
2	-7.2328603802162\\
2.25	-8.48976859706488\\
2.5	-9.30323389422674\\
2.75	-9.8229967480628\\
3	-10.5475353285066\\
3.25	-11.340297025477\\
3.5	-12.4270125978456\\
3.75	-14.2286510743705\\
4	-14.8009099557291\\
4.25	-15.4695326998749\\
4.5	-16.2662451723617\\
4.75	-17.138850053018\\
5	-18.3242400648828\\
};
\addlegendentry{Real Cube Split};

\addplot [color=blue,solid,mark=o,mark options={solid}]
  table[row sep=crcr]{%
0.75	-3.86733016206536\\
1	-4.60115458620276\\
1.25	-5.26775598793606\\
1.5	-6.07461676384649\\
1.75	-7.23911967003172\\
2	-8.48145627043766\\
2.25	-9.0220134088468\\
2.5	-9.82882334964188\\
2.75	-10.4168185661119\\
3	-11.3543521416071\\
3.25	-12.3245403518324\\
3.5	-13.9598667579072\\
3.75	-14.5301305249676\\
4	-15.1720066746092\\
4.25	-15.9970476755129\\
4.5	-16.8335672675803\\
4.75	-17.9876157737229\\
5	-19.7322716452255\\
};
\addlegendentry{Complex Cube Split};

\addplot [color=mycolor1,solid,mark=diamond,mark options={solid}]
  table[row sep=crcr]{%
0.75	-3.1727718796401\\
0.75	-3.1727718796401\\
1.5	-6.46665846436044\\
1.5	-6.46665846436044\\
1.5	-6.46665846436044\\
1.5	-6.46665846436044\\
2.5	-9.9563391591441\\
2.5	-9.9563391591441\\
2.5	-9.9563391591441\\
2.5	-9.9563391591441\\
3.5	-13.9453092410993\\
3.5	-13.9453092410993\\
3.5	-13.9453092410993\\
3.5	-13.9453092410993\\
4.5	-17.5011164717623\\
4.5	-17.5011164717623\\
4.5	-17.5011164717623\\
4.5	-17.5011164717623\\
};
\addlegendentry{SLAQ};

\addplot [color=green,solid,mark=square,mark options={solid}]
  table[row sep=crcr]{%
0.75	-3.75119944738142\\
1	-4.07260697369796\\
1.25	-4.13095378584254\\
1.5	-4.05317042026725\\
1.75	-4.18268152990292\\
2	-4.08048200938786\\
2.25	-4.16460138956458\\
2.5	-4.15040086555292\\
2.75	-4.0469361588836\\
3	-4.17588034663501\\
3.25	-4.18085006841128\\
3.5	-4.14567273647966\\
3.75	-4.17048291597988\\
4	-4.10963500503297\\
4.25	-4.21195381854892\\
4.5	-4.05423831420698\\
4.75	-4.22138049925651\\
5	-4.04226779918432\\
};
\addlegendentry{Fourier codebook};

\addplot [color=mycolor2,solid,mark=+,mark options={solid}]
  table[row sep=crcr]{%
0.75	-2.29005403862224\\
1	-2.91972990666929\\
1.25	-3.34043191704684\\
1.5	-3.92004002693772\\
1.75	-4.67628081969982\\
2	-5.64321883289022\\
2.25	-5.82453089103109\\
2.5	-6.05102263388441\\
2.75	-6.30498665442415\\
3	-6.62951454797845\\
3.25	-7.07689184336418\\
3.5	-7.67430249811245\\
3.75	-8.44671743712009\\
4	-9.51196802043814\\
4.25	-9.90127225487692\\
4.5	-10.3295579457468\\
4.75	-10.8251389636678\\
5	-11.3843310884405\\
};
\addlegendentry{Scalar quantization};

\addplot [color=black,solid]
  table[row sep=crcr]{%
0.75	-5.50907468880581\\
1	-6.51250800768575\\
1.25	-7.51594132656568\\
1.5	-8.51937464544562\\
1.75	-9.52280796432556\\
2	-10.5262412832055\\
2.25	-11.5296746020854\\
2.5	-12.5331079209654\\
2.75	-13.5365412398453\\
3	-14.5399745587252\\
3.25	-15.5434078776052\\
3.5	-16.5468411964851\\
3.75	-17.5502745153651\\
4	-18.553707834245\\
4.25	-19.5571411531249\\
4.5	-20.5605744720049\\
4.75	-21.5640077908848\\
5	-22.5674411097647\\
};
\addlegendentry{Lower bound of optimal codebook};

\addplot [color=black,dashed]
  table[row sep=crcr]{%
0.75	-3.99347006474782\\
1	-4.99690338362776\\
1.25	-6.0003367025077\\
1.5	-7.00377002138764\\
1.75	-8.00720334026757\\
2	-9.01063665914751\\
2.25	-10.0140699780274\\
2.5	-11.0175032969074\\
2.75	-12.0209366157873\\
3	-13.0243699346673\\
3.25	-14.0278032535472\\
3.5	-15.0312365724271\\
3.75	-16.0346698913071\\
4	-17.038103210187\\
4.25	-18.0415365290669\\
4.5	-19.0449698479469\\
4.75	-20.0484031668268\\
5	-21.0518364857068\\
};
\addlegendentry{Upper bound of optimal codebook};

\end{axis}
\end{tikzpicture}%
\caption{Average chordal quantization error (in dB) vs. number of bits per dimension for codebooks on $G(\mathbb{C}^4,1)$.}
\label{fig:gain}
\end{figure}
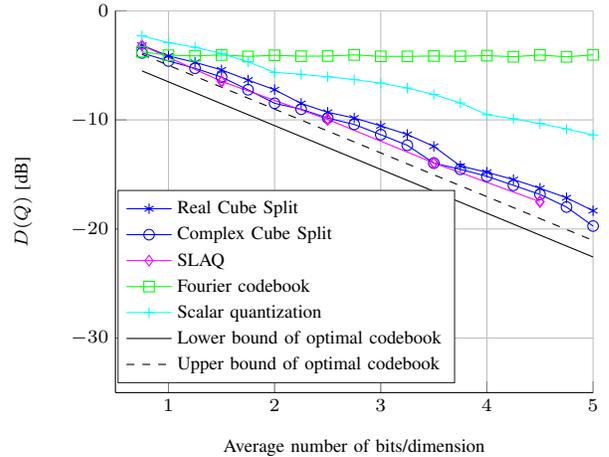
\begin{figure}[h!]
\scriptsize 
\centering
%
%
\definecolor{mycolor1}{rgb}{1.00000,0.00000,1.00000}%
\definecolor{mycolor2}{rgb}{0.00000,1.00000,1.00000}%
\begin{tikzpicture}

\begin{axis}[%
width=2.5in,
height=2in,
at={(0.758333in,0.48125in)},
scale only axis,
xmin=0,
xmax=5,
xmajorgrids,
ymin=-25,
ymax=0,
ymajorgrids,
axis x line*=bottom,
axis y line*=left,
legend style={at={(0.035119,0.026984)},anchor=south west,legend cell align=left,align=left,draw=white!15!black},
xlabel={Average number of bits/dimension},
ylabel={$D(Q)$ [dB]},
]
\addplot [color=blue,solid,mark=asterisk,mark options={solid}]
  table[row sep=crcr]{%
0.109375	-0.320450213712694\\
0.359375	-0.703678860416021\\
0.609375	-1.11313713950296\\
0.859375	-1.55991424455248\\
1.109375	-2.06080273096361\\
1.359375	-2.60656911356906\\
1.609375	-3.2645922440599\\
1.859375	-4.01647931446754\\
2.109375	-4.85136039295606\\
2.359375	-5.1307546593623\\
2.609375	-5.50320024453565\\
2.859375	-5.8947880250819\\
3.109375	-6.38577718139816\\
3.359375	-6.92376465442701\\
3.609375	-7.66229294875195\\
3.859375	-8.52205089526123\\
4.109375	-9.52437865095296\\
4.359375	-9.95776199292568\\
4.609375	-10.3767640739324\\
4.859375	-10.9122005245768\\
};
\addlegendentry{Real Cube Split};

\addplot [color=blue,solid,mark=o,mark options={solid}]
  table[row sep=crcr]{%
0.109375	-0.358620568467307\\
0.359375	-0.738056622395871\\
0.609375	-1.15715940936568\\
0.859375	-1.61551442965919\\
1.109375	-2.12807432896411\\
1.359375	-2.6768131827014\\
1.609375	-3.34240405009271\\
1.859375	-4.08886479464164\\
2.109375	-4.88550239070268\\
2.359375	-5.21678737143919\\
2.609375	-5.63185858118314\\
2.859375	-6.03492366699328\\
3.109375	-6.55597052327053\\
3.359375	-7.16355031561129\\
3.609375	-7.92875495187678\\
3.859375	-8.81484440615869\\
4.109375	-9.73983985964896\\
4.359375	-10.1735786511541\\
4.609375	-10.5801135460248\\
4.859375	-11.1252674851873\\
};
\addlegendentry{Complex Cube Split};

\addplot [color=mycolor1,solid,mark=diamond,mark options={solid}]
  table[row sep=crcr]{%
0.984375	-1.75209084784888\\
0.984375	-1.75209084784888\\
0.984375	-1.75209084784888\\
0.984375	-1.75209084784888\\
1.96875	-4.7151252572437\\
1.96875	-4.7151252572437\\
1.96875	-4.7151252572437\\
1.96875	-4.7151252572437\\
2.96875	-6.40623287924511\\
2.96875	-6.40623287924511\\
2.96875	-6.40623287924511\\
2.96875	-6.40623287924511\\
3.96875	-9.91261132576855\\
3.96875	-9.91261132576855\\
3.96875	-9.91261132576855\\
3.96875	-9.91261132576855\\
4.96875	-11.3208802079062\\
4.96875	-11.3208802079062\\
4.96875	-11.3208802079062\\
4.96875	-11.3208802079062\\
};
\addlegendentry{SLAQ};

\addplot [color=green,solid,mark=square,mark options={solid}]
  table[row sep=crcr]{%
0.109375	-0.349312883638226\\
0.359375	-0.368308457077387\\
0.609375	-0.377228391466411\\
0.859375	-0.376091246674185\\
1.109375	-0.373970020399869\\
1.359375	-0.373562329416528\\
1.609375	-0.375380282968358\\
1.859375	-0.3733590735307\\
2.109375	-0.37454733759943\\
2.359375	-0.373037975210222\\
2.609375	-0.369107549285097\\
2.859375	-0.374177277377995\\
3.109375	-0.369019155136976\\
3.359375	-0.374185773954272\\
3.609375	-0.374945557711173\\
3.859375	-0.373681373129471\\
4.109375	-0.37199723872796\\
4.359375	-0.370069585533006\\
4.609375	-0.375135015309956\\
4.859375	-0.371347284418894\\
};
\addlegendentry{Fourier codebook};

\addplot [color=mycolor2,solid,mark=+,mark options={solid}]
  table[row sep=crcr]{%
0.109375	-0.202020898042328\\
0.359375	-0.580001290314166\\
0.609375	-0.994410547644004\\
0.859375	-1.45224030845501\\
1.109375	-1.95637534141322\\
1.359375	-2.52410869297773\\
1.609375	-3.18004363808542\\
1.859375	-3.95639724903066\\
2.109375	-4.36291063061747\\
2.359375	-4.14879376101921\\
2.609375	-3.96955484654625\\
2.859375	-3.82627936122136\\
3.109375	-3.72459277549704\\
3.359375	-3.68849854400316\\
3.609375	-3.76669436689694\\
3.859375	-4.06812735089064\\
4.109375	-4.37226638663072\\
4.359375	-4.18411391902271\\
4.609375	-4.02311106296832\\
4.859375	-3.9018247180109\\
};
\addlegendentry{Scalar quantization};

\addplot [color=black,solid]
  table[row sep=crcr]{%
0.109375	-0.471266263566089\\
0.359375	-1.23578688747461\\
0.609375	-2.00030751138314\\
0.859375	-2.76482813529166\\
1.109375	-3.52934875920018\\
1.359375	-4.29386938310871\\
1.609375	-5.05839000701723\\
1.859375	-5.82291063092575\\
2.109375	-6.58743125483428\\
2.359375	-7.3519518787428\\
2.609375	-8.11647250265132\\
2.859375	-8.88099312655985\\
3.109375	-9.64551375046837\\
3.359375	-10.4100343743769\\
3.609375	-11.1745549982854\\
3.859375	-11.9390756221939\\
4.109375	-12.7035962461025\\
4.359375	-13.468116870011\\
4.609375	-14.2326374939195\\
4.859375	-14.997158117828\\
};
\addlegendentry{Lower bound of optimal codebook};

\addplot [color=black,dashed]
  table[row sep=crcr]{%
0.109375	-0.412273083760021\\
0.359375	-1.17679370766855\\
0.609375	-1.94131433157707\\
0.859375	-2.70583495548559\\
1.109375	-3.47035557939412\\
1.359375	-4.23487620330264\\
1.609375	-4.99939682721116\\
1.859375	-5.76391745111969\\
2.109375	-6.52843807502821\\
2.359375	-7.29295869893673\\
2.609375	-8.05747932284526\\
2.859375	-8.82199994675378\\
3.109375	-9.58652057066231\\
3.359375	-10.3510411945708\\
3.609375	-11.1155618184794\\
3.859375	-11.8800824423879\\
4.109375	-12.6446030662964\\
4.359375	-13.4091236902049\\
4.609375	-14.1736443141134\\
4.859375	-14.938164938022\\
};
\addlegendentry{Upper bound of optimal codebook};

\end{axis}
\end{tikzpicture}%
\caption{Average Chordal Error (in dB) vs. number of bits per dimension for codebooks on $G(\mathbb{C}^{64},1)$.}
\label{fig:gain2}
\end{figure}
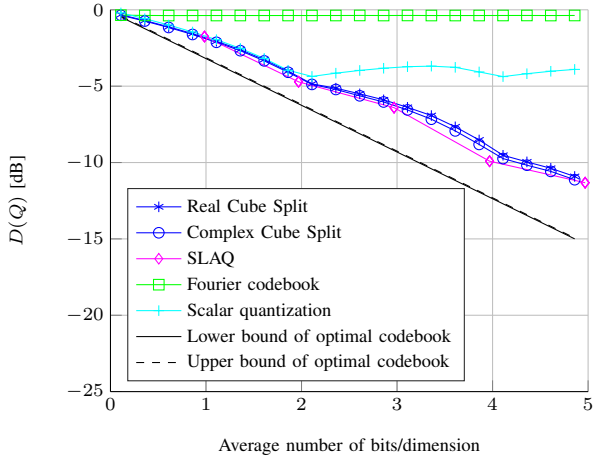

\begin{table}
\caption{Feature summary of the different quantization approaches in the high-resolution (large $D$ and $B$)  regime.}
\label{table:feature}
\begin{tabular}{|p{2cm}|p{3cm}|p{2.5cm}|}
\hline
Method & Encoding complexity & Efficiency\\
\hline\hline
Fourier Codebook \cite{boccardi07}& polynomial in $D$ and $B$ & Tends to 0 for large codebooks\\
\hline
SLAQ \cite{ryan07}& exponential in $B$; polynomial in $D$ & high\\
\hline
Scalar quantization & linear in $D$; independent from $B$ & low\\
\hline
Unstructured codebook & exponential in $B$ & high\\
\hline
Cube Split quantizer & linear in $D$; independent from $B$ & high\\
\hline
\end{tabular}
\end{table}

\section{Conclusions}
In this paper we addressed the problem of quantization of a Grassmannian element for beamforming applications. We proposed a new codebook with low-complexity encoder/decoder allowing a fast quantization even for high-dimensional, high-resolution quantization applications. The performance of the new codebook in terms of distortion is on par with the best state of the art quantization methods with a complexity equivalent to the simple scalar quantization.\\

\appendices

\section{Proofs of Propositions \ref{prop1} and \ref{prop2}} 
\label{section_appendixproofs}

\begin{proposition}[Real case]
\label{prop1}
Let $\y\in\mathbb{R}^d$ be a random vector uniformly distributed on $C_{i^*}^0$ for an arbitrary initial cell index $i^*$; using the compander defined by eqs.~\eqref{real_compander_1} and \eqref{eq:mapping}, the resulting $\M_i(\y)$ is a random vector on $[0;1]^{d-1}$ with uniform marginals.
\end{proposition}
\begin{proof}
Let $\y=(y_1,\dots,y_d)^T$ be a random vector uniformly distributed on the unit sphere of $\mathbb{R}^d$. Then, if $i\neq i^*$, the random variable $t_i=\frac{y_i}{y_{i^*}}$ is drawn from a Cauchy distribution whose cumulative distribution function (cdf) is given by $x\mapsto \frac{1}{\pi}\tan^{-1}(x)+\frac{1}{2}$. The belonging to the $i^*$-th initial cell $C_{i^*}^0$ imposes $\left|t_i\right|<1$. It is then clear that the $i$-th component of $\M_{i^*}$ corresponds to the cdf of the restriction of the Cauchy distribution to $[-1;1]$.\\
\end{proof}

\begin{proposition}[Complex case]
\label{prop2}
Let $\y\in \mathbb{C}^D$ be a random vector uniformly distributed on the $i$-th initial cell of the complex Grassmannian $C^{\mathbb{C}}_i$; using the complex compander defined by eqs.~\eqref{eq:mapping2_0}--\eqref{eq:mapping2}, the resulting $\M^{\mathbb{C}}_i(\y)$ is a random vector on $[0;1]^{2D-2}$ with uniform marginals.
\end{proposition}
\begin{proof}
We will prove that $w_i$ defined by Eq. (\ref{eq:mapping2}) follows a Gaussian distribution. Let us first note that $\frac{t_i}{|t_{i}|}$ and $|t_i|$ are independent since $t_i$ follows a complex Cauchy distribution which is complex elliptical (see Th.~4 of \cite{ollila12}). Thus, it suffices to prove that $|w_i|=\sqrt{2}\log\left(\frac{1+|t_i|^2}{1-|t_i|^2}\right)^{1/2}$ is Rayleigh-distributed.\\
On an other hand, since $|t_i|^2$ may be seen as the quotient of two independent $\chi^2$ random variables, its distribution is a Fisher$(2,2)$ truncated on $[0;1]$ whose cdf is $F:x\mapsto \frac{2x}{x+1}$. Therefore, denoting  the quantile of the Rayleigh distribution as $Q:t\mapsto \sqrt{2}\log\left(\frac{1}{1-t}\right)^{1/2}$, it holds that $|w_i|=Q(F(|t_i|^2))$.
\end{proof}

\balance
\bibliographystyle{IEEEtran}
\bibliography{refen2}
\end{document}